\documentclass[sigconf]{acmart}
\AtBeginDocument{%
	}

\setcopyright{acmlicensed}
\makeatletter
\def\ACM@cc@type{by}   
\def\@ACM@license@logo@pdf{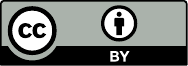}  
\makeatother
\copyrightyear{2026}
\acmYear{2026}
\acmDOI{xxxx}
\acmConference[Conference acronym 'Internetware2026]{the 17th International Conference on Internetware}{July 18--20,
	2026}{Gold Coast, Australia}
\acmISBN{978-1-4503-XXXX-X/2018/06}

\usepackage{algorithmic}
\usepackage{algorithm}

\usepackage{subfigure}

\usepackage{cleveref}
\Crefname{figure}{fig}{figures}
\Crefname{figure}{Fig}{Figures}
\Crefname{table}{Table}{Tables}
\Crefname{alg}{Algorithm}{Algorithms}

\usepackage{enumitem}
\usepackage{listings}

\newcommand{\Com}[1]{\textcolor{blue!70!black}{{#1}}}

\begin{document}
	
	\title{TRM-Raft: A Byzantine-Resistant Raft Consensus via Integrated Trust and Reputation Model}
	
	
	
	\author{Jie Zhang}
	\authornote{Jie Zhang and Xubo Fan contributed equally to this research.\\
		Xiaohongli is the corresponding author.}
	\orcid{0000-0001-8476-0215}
	\affiliation{%
		\institution{Tianjin University}
		\city{Tianjin}
		\country{China}
	}
	\email{jackzhang@tju.edu.cn}
	
	\author{Xubo Fan}
	\authornotemark[1]
	\affiliation{%
		\institution{Tianjin University}
		\city{Tianjin}
		\country{China}
	}
	\email{xubofan@tju.edu.cn}

	\author{Xiaohong Li}
	\correspondingauthor
	\affiliation{%
		\institution{Tianjin University}
		\city{Tianjin}
		\country{China}
	}
	\email{xiaohongli@tju.edu.cn}
	
	\author{Zhiyong Feng}
	\affiliation{%
		\institution{Tianjin University}
		\city{Tianjin}
		\country{China}
	}
	\email{zyfeng@tju.edu.cn}

	\renewcommand{\shortauthors}{Zhang et al.}
	
	\begin{abstract}

		Internetware envisions autonomous software entities dynamically collaborating over the open, evolving Internet. A key enabler of such systems is the Raft consensus protocol, which is widely adopted for its simplicity and high performance in distributed coordination, e.g., in service registries, cloud orchestration, and permissioned blockchains. However, Raft’s design assumes a purely crash-fault model, making it inherently vulnerable to Byzantine behaviors such as election forgery and log tampering when deployed in the hostile, dynamic Internet environment. Existing Byzantine fault-tolerant protocols incur prohibitive communication overhead and invasive architectural changes, while point-wise hardening attempts for Raft fail to provide unified, adaptive defense.
		
		In this paper, we propose \textbf{TRM-Raft}, a Byzantine-resistant enhancement of Raft that \textit{non-intrusively} integrates a Blockchain-based Trust and Reputation Model (B-TRM) into the consensus core. The core idea is to quantify multi-dimensional node behaviors, apply adaptive penalties that distinguish accidental faults from persistent malice, and embed reputation signals directly into leader election and log replication. Specifically, a reputation-aware election mechanism detects and harshly penalizes term/index forgery, keeping low-reputation nodes out of the leader role. A Schnorr-signature-based leader restriction mechanism enables followers to instantly verify log integrity; any tampering triggers reputation decay and leader replacement. Implemented and evaluated in a realistic Internetware setting using Hyperledger Fabric, TRM-Raft maintains a malicious leader ratio below 5\% even when 40\% of nodes are Byzantine, while introducing less than 10\% throughput degradation and under 5\% latency increase compared to vanilla Raft. TRM-Raft thus provides a lightweight and practical path toward trustworthiness for the broad class of Internetware systems that rely on Raft.
	\end{abstract}



	
	\maketitle

\section{Introduction}
\label{Sec 1}

The Internetware paradigm envisions autonomous software entities dynamically collaborating over the open, ever-changing Internet~\cite{DBLP:conf/internetware/ChenZYZSZ20,10.1145/3755881.3755927}. Modern Internet-scale distributed systems, from cloud orchestration platforms and service registries to permissioned enterprise blockchains \cite{DBLP:journals/tsc/ZhangLZFXHB26,DBLP:conf/cscwd/ZhangLZHBF25,DBLP:conf/tase/MaoZZL25}, critically rely on consensus protocols to maintain a consistent and reliable system state in this dynamic environment\cite{DBLP:journals/comsur/XiaoZLH20,DBLP:journals/csur/00310J23}. Among these protocols, the Raft consensus algorithm~\cite{DBLP:conf/usenix/OngaroO14} has gained widespread adoption due to its simplicity, understandability, and high performance. Prominent software infrastructures such as etcd (the backbone of Kubernetes)\footnote{Raft is the core of distributed container consensus algorithms.}, Quorum\footnote{Raft is one of two consensus protocols}, and Hyperledger Fabric\footnote{Supports the Raft consensus starting from version 1.4.1}~\cite{Androulaki2018} all leverage Raft to coordinate replicas.

However, Raft's design assumes a purely Crash Fault Tolerant (CFT) model, where nodes may only stop working but never deviate from the protocol~\cite{DBLP:conf/usenix/OngaroO14}. In the open Internet environment that Internetware systems inhabit, this non-Byzantine assumption is fragile~\cite{DBLP:conf/usenix/OngaroO14,DBLP:conf/icbc2/BellajOBCM22}. Nodes can be compromised by external attackers, exhibit selfish behaviors, or become faulty in adversarial ways, thereby threatening the very consistency and availability that consensus is meant to provide~\cite{WANG2022,RaftvotingMechComp, ZHETU2022109404,DBLP:journals/tosn/TianBSZG24}. Two particularly severe classes of Byzantine behavior against Raft are:


\textbf{Forgery Attacks on Leader Election}
In Raft, leader election is predicated on candidates possessing the “latest” log, determined by term number ($Term$) and log index ($Index$). A malicious node can exploit this by \textit{forging} artificially high $Term$ and $Index$ values. By minimizing its election timeout ($ET$), the attacker can transition to candidate earlier than honest nodes, broadcast a \texttt{RequestVote} RPC with forged metadata, and win the election before legitimate candidates can react. Once elected, the malicious leader can suppress further elections via heartbeats, even if later replaced, it can repeat the attack by incrementing its term again. This vulnerability, highlighted by \cite{WANG2022} and \cite{RaftvotingMechComp}, allows a single malicious node to persistently disrupt leader legitimacy.

\textbf{Tampering Attacks on Log Replication}
After election, the leader is solely responsible for replicating client requests (logs) to followers. Raft ensures log consistency through term and index matching but does \textit{not} verify the integrity of the log content ($entries$). A malicious leader can therefore \textit{tamper} with $entries$ while preserving correct $Term$ and $Index$ structure, bypassing Raft’s native consistency checks. As demonstrated in \cite{RaftSignCompareSRaft,ZHETU2022109404, RaftSigntCompareRBraft,DBLP:journals/tosn/TianBSZG24}, such tampering can corrupt the blockchain state, compromise data integrity, and undermine trust in the system.


While Byzantine Fault Tolerant (BFT) protocols like PBFT~\cite{DBLP:conf/osdi/CastroL99} can inherently resist such adversaries, their high communication complexity ($O(N^2)$) and demanding quorum requirements $(n=3f+1)$ introduce substantial performance penalties and architectural intrusiveness~\cite{DBLP:journals/tc/ChenEMH23}. For many Internetware systems, completely replacing an existing Raft stack with a full BFT protocol is both costly and impractical~\cite{DBLP:conf/internetware/ChenZYZSZ20,DBLP:conf/internetware/MiaoZCYZJZ22,11534385,JieZHANG_2104}. Existing defenses that retrofit Raft with static thresholds~\cite{WANG2022} or hash-chain verification~\cite{ZHETU2022109404} focus on isolated attack vectors, lack a unified trust management framework, and often fail to adapt to the dynamic and evolving threat landscape of the Internet.

To bridge this gap, we propose \textbf{TRM-Raft}, a Byzantine-Resistant Raft consensus achieved by integrating a blockchain-based Trust and Reputation Model (TRM) directly into the consensus core. Our key insight is that malicious behaviors in Internetware systems, whether forgery, tampering, or other camouflaged attacks (e.g., On-Off attacks\cite{DBLP:conf/blockchain2/CamiloRSD20}, manifest as observable anomalies that can be quantified and penalized through a dynamic reputation system, \textit{without redesigning the fundamental Raft protocol}. TRM-Raft transforms the static trust assumption of Raft into an adaptive, behavior-aware trust management while preserving its proven efficiency and simplicity.
By co-designing a behavioral trust model with cryptographic verification, TRM-Raft is, to the best of our knowledge, the first Raft variant that unifies dynamic reputation-based election gating with Schnorr-signature-verified log integrity. This combination forces an adversary to simultaneously defeat anomaly detection and signature forgery, a synergy absent from prior point-wise defenses.

Our key contributions are:
\begin{enumerate}
\item \textbf{Unified Trust and Integrity Framework:} The first integration of a multi-dimensional reputation model (B-TRM) with Schnorr signatures in Raft, jointly addressing election forgery and log tampering through adaptive behavioral assessment and cryptographic verification.
\item \textbf{Reputation-Aware Election Mechanism:} A dynamic monitoring scheme that detects anomalous term/index surges during elections, penalizes forgery attempts by halving the attacker's reputation, and excludes low-reputation nodes from voting and candidacy, thus guaranteeing that leadership remains in trustworthy hands.
\item \textbf{Leader Restriction with Schnorr Signatures:} An efficient cryptographic verification layer embedded in the log replication phase. Followers verify log integrity via Schnorr signatures; any tampering leads to immediate reputation penalties and, if sustained, triggers leader replacement, ensuring end-to-end state trustworthiness.
\item \textbf{Practical Validation:} We implement TRM-Raft on Hyperledger Fabric, a representative Internet-scale permissioned platform, and demonstrate that it maintains malicious leader prevalence below 5\% even when 40\% of nodes are Byzantine, while incurring modest performance overhead (less than 10\% in throughput and 5\% in latency) compared to vanilla Raft.
\end{enumerate}

The remainder of this paper is organized as follows: \Cref{Sec 2} reviews related work. \Cref{Sec 3} formalizes the threat model for Internetware-Raft systems. \Cref{Sec 4} presents the B-TRM reputation model. \Cref{Sec 5} describes TRM-Raft's design and provides a security analysis. \Cref{Sec 6} gives experimental evaluation, \Cref{sec:limitations} discusses limitations, and \Cref{Sec 7} concludes.

\section{Related Works}
\label{Sec 2}

\subsection{Raft, Byzantine behavior, and modern BFT baselines}
Raft~\cite{DBLP:conf/usenix/OngaroO14} is a leader-based crash-fault tolerant consensus protocol, designed with clarity, implementability, and predictable performance in mind, which has established it as the consensus backbone for numerous Internetware infrastructures~\cite{DBLP:journals/comsur/XiaoZLH20}. In benign deployments, its majority quorum and one-round commit path ($\lfloor N/2 \rfloor + 1$ nodes for committing log entries\cite{DBLP:conf/icbc2/BellajOBCM22}) are sufficient to support high throughput with relatively low coordination overhead. However, its vulnerability to Byzantine participants directly threatens the reliability of many platforms. Fundamentally, it assumes that participants adhere strictly to the protocol, barring crash-stop failures. Once a node can lie about its election state~\cite{WANG2022, RaftvotingMechComp}, equivocate on log progress~\cite{ZHETU2022109404, RaftSigntCompareRBraft}, or tamper with replicated content~\cite{RaftSignCompareSRaft,DBLP:journals/tosn/TianBSZG24}, Raft's native checks are no longer sufficient.

Classical Byzantine fault tolerant protocols such as PBFT \cite{DBLP:conf/osdi/CastroL99} established the standard $3f+1$ resilience threshold (at least $\lfloor 2N/3 \rfloor + 1$ nodes to tolerate up to $f$ malicious nodes out of $N=3f+1$\cite{DBLP:conf/icbc2/BellajOBCM22}), while more recent protocols such as HotStuff \cite{DBLP:conf/podc/YinMRGA19} and modern blockchain implementations such as CometBFT\footnote{\url{https://github.com/cometbft/cometbft.git}} reduce practical overhead but still preserve quorum-based Byzantine safety~\cite{DBLP:journals/tc/ChenEMH23}. These systems are stronger than Raft under an arbitrary Byzantine adversary, but their message patterns, view-change logic, and recovery machinery are more expensive than CFT protocols~\cite{DBLP:journals/csur/00310J23}. The design space therefore remains split between low-cost CFT and high-cost BFT~\cite{RaftvotingMechComp, DBLP:journals/tosn/TianBSZG24}. TRM-Raft is positioned between the two: it aims to suppress a restricted but practically relevant set of observable Byzantine deviations without claiming equivalence to strict BFT.

\subsection{TRMs in distributed systems}
Trust and Reputation Models (TRMs) have been used to score participants based on past behavior, local observations, and indirect evidence \cite{ZHETU2022109404,11164303,HUANGMinmin_733,PoR_Gai}. In distributed systems, such models are most useful when the system can observe repeated actions, correlate them across peers, and enforce penalties that are visible to the rest of the network~\cite{11164303,HUANGMinmin_733}. This makes them attractive for permissioned consensus settings, where node identities are known and the protocol already maintains a shared state.

At the same time, reputation mechanisms are not a substitute for cryptographic safety~\cite{11164303}. They are reactive, can be gamed, and may fail against a strategic adversary that behaves honestly long enough to accumulate trust before misbehaving~\cite{DBLP:conf/icbc2/MakhdoomTZAL20,RaftvotingMechComp}. Prior work on accountability and alternative fault models emphasizes this point and suggests that trust-based mechanisms should be framed as partial defenses rather than as full Byzantine guarantees~\cite{DBLP:conf/icdcs/CivitGG21}. Our design follows that principle.


\subsection{Cryptographic integrity for replicated logs}
Digital signatures are a standard method for binding a message to its origin and contents. Schnorr signatures \cite{SchnorrSig} are attractive because they are compact, efficient to verify, and well suited to modern blockchain-style implementations \cite{DBLP:journals/dcc/MaxwellPSW19}. 
In the context of hardening Raft against tampering attacks, Schnorr signatures provide a cryptographic mechanism for log integrity verification\cite{DBLP:conf/icuimc/TianLZZ21}. By producing a Schnorr signature in the propose phase, followers can verify the signature against the known public key to detect the manipulation instantly. This directly counters the tampering attack described in \cite{ZHETU2022109404,DBLP:journals/tosn/TianBSZG24,DBLP:conf/icuimc/TianLZZ21}, where a malicious leader alters log content while preserving correct $Term$ and $Index$ metadata.

In TRM-Raft, signatures do not replace consensus; they protect the integrity of replicated entries so that a leader cannot silently alter payloads after a request has been endorsed or accepted for replication. The signature check therefore complements the reputation mechanism: one detects repeated protocol abuse, the other prevents content tampering from being committed.

\subsection{Research gap}
Existing work typically addresses either election forgery or log tampering in isolation, or assumes a general BFT protocol, leaving a critical gap for distributed Internet systems requiring stronger than CFT resilience with transparent enforcement and limited overhead. TRM-Raft targets this gap by providing a practical, deployable mechanism through the integration of TRM into Raft's core phases, leader election and log replication, to unify defenses against multiple attack vectors. While TRMs have been explored in broader blockchain contexts \cite{DBLP:conf/icbc2/MakhdoomTZAL20,RaftvotingMechComp} and BFT protocols \cite{PoR_Gai,11164303}, their tailored adaptation to Raft's specific workflow, combined with cryptographic log binding, constitutes our novel contribution.
%
%
%

\section{Threat Model and Attack Analysis}
\label{Sec 3}

We consider an Internetware system where multiple replicas coordinate via Raft, and where nodes may exhibit \textit{observable Byzantine behaviors}.
Crucially, we do not claim to tolerate arbitrary, covert Byzantine faults (as in full BFT). Instead, we define a \textit{Byzantine-Resistant} fault model targeting a specific set of deviations that can be detected via anomaly monitoring or cryptographic verification.

\subsection{Defined Failure Set}
Our model handles the following attack classes:
\begin{enumerate}
\item \textbf{Forgery Attacks:} Nodes claiming artificially inflated $Term$ or $Index$ values during leader election \cite{RaftvotingMechComp}.
\item \textbf{Tampering Attacks:} Leaders altering the payload ($m$) of a log entry while preserving metadata integrity \cite{DBLP:journals/tosn/TianBSZG24}.
\item \textbf{On-Off Attacks:} Nodes alternating between honest and malicious behavior to evade detection \cite{DBLP:conf/blockchain2/CamiloRSD20}.
\end{enumerate}
\noindent\textbf{Observability and Reporting Assumption.}
All attacks in the defined failure set are assumed to be \textit{observable}: forgery manifests as anomalous term/index surges detectable by the election monitor; tampering is detected by Schnorr signature verification. We also assume that the reputation reporting subsystem receives inputs from multiple independent peers, and that a simple majority of reporting nodes for any metric is honest. This assumption prevents malicious nodes from indefinitely suppressing an honest node's reputation score. We explicitly do not handle attacks where a leader reorders transactions without altering content (equivocation), which is a limitation discussed in Section~\ref{sec:limitations}.

We assume the adversary cannot break standard cryptographic primitives (Schnorr signatures, hash functions) and that the network is partially synchronous.

\subsection{Forgery Attack Formalization}
A forgery attack occurs during Raft’s leader election phase. Let $ET_i \in [150\text{ms}, 300\text{ms}]$ be the election timeout of node $i$, randomly assigned. A malicious node $M$ sets $ET_M = 150\text{ms}$ (minimum). When the current leader fails, $M$ quickly transitions to candidate, increments its term to $Term_M^t = Term_M^{t-1} + \Delta$, where $\Delta$ is a forged large value (e.g., doubling the previous term), and broadcasts $\texttt{RequestVote}$ $(UID_M, Term_M^t, Index_M^t)$. Honest followers, upon receiving this message with a higher $Term$, are compelled to vote for $M$ (per Raft rules \cite{DBLP:conf/usenix/OngaroO14}), enabling $M$ to win the election despite lacking legitimate log progress. A simple Forgery Attack method is illustrated in \Cref{fig:signature_flow}.

\begin{figure}[htb]
\centering
\includegraphics[width=\linewidth]{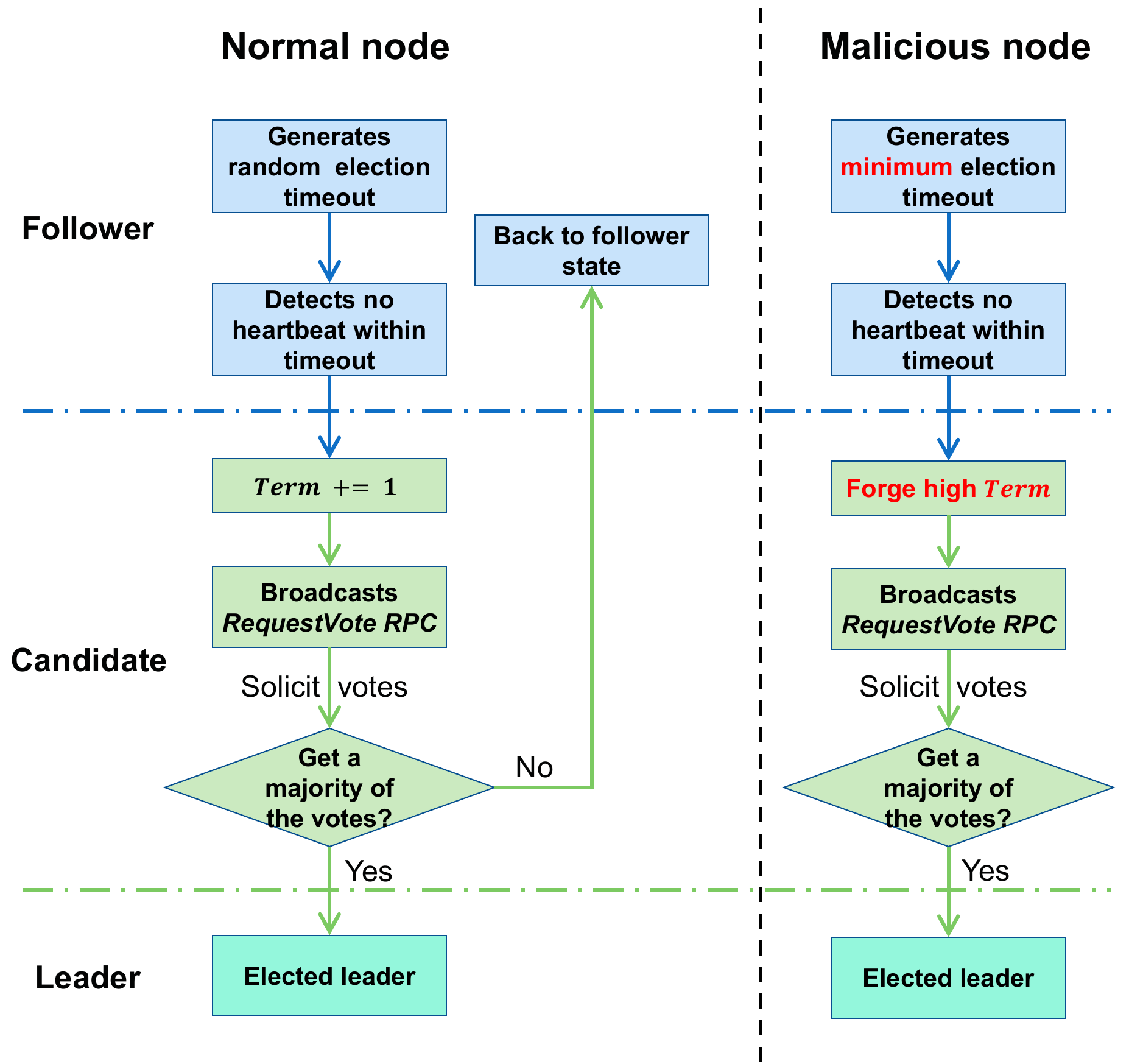}
\caption{A case of forgery attack.}
\label{fig:signature_flow}
\end{figure}

\subsection{Tampering Attack Formalization}
Once elected, leader $L$ receives a client request $m$, which it should append as a log entry $e = (Term, Index, m)$. A malicious $L$ can tamper $m$ to $m'$ while preserving $Term$ and $Index$. It then broadcasts $\texttt{AppendEntries}(Term, Index, m')$ to followers. Since Raft only checks $Term$ and $Index$ for consistency, followers accept the tampered entry, corrupting the replicated state machine.

\subsection{Additional Attack Vectors}
We also consider \textit{On-Off attacks}, where nodes alternate between honest and malicious behavior to evade detection \cite{DBLP:conf/blockchain2/CamiloRSD20}, and \textit{Sybil attacks}\cite{DBLP:conf/iptps/Douceur02}, where an adversary creates multiple identities to influence consensus. 
TRM-Raft mitigates these attacks through a combination of identity registration (Section~\ref{subsec:registration}) and the reputation penalty rules defined in Section~\ref{Sec 4}.


\section{B-TRM: A Lightweight Blockchain-Based TRM}
\label{Sec 4}

To enable Raft consensus to distinguish between honest and malicious nodes in a Byzantine environment, we design a TRM tailored for blockchain-based data-sharing platforms. Our model, named B-TRM (blockchain-based TRM for Real-World Scenarios), operates on-chain and quantifies node behavior through multi-dimensional assessment. Unlike prior TRMs \cite{ZHETU2022109404,RaftvotingMechComp,11164303} that focus solely on attack detection, B-TRM introduces adaptive penalty rules and a dynamic evaluation mechanism to differentiate between malicious attackers and occasional mistake-making users—a critical distinction in practical deployments.

\subsection{Model Formalization}
B-TRM is defined as a six-tuple:
\[
\text{B-TRM} = \langle U, B, F, W, P, \Phi \rangle
\]
where:
\begin{itemize}
\item $U = \{u_1, \dots, u_n\}$ is the set of users, each represented as $u_i = (uid_i, Rep_i)$ where $Rep_i \in [0,1]$ denotes reputation.
\item $B$ is the set of possible actions, categorized into \textit{success} ($s^x$) and \textit{failure} ($f^x$) types for three behavior classes $x \in \{$ \texttt{Upload}, \texttt{Modify}, \texttt{Loss}$\}$.
\item $F: B^* \rightarrow \mathbb{R}$ is the reputation evaluation function that maps a user’s action sequence to a reputation score.
\item $W$ is a dynamic adjustment function that sets the reputation evaluation interval based on $Rep_i$.
\item $P$ is a set of penalty rules that trigger temporary or permanent freezing of low-reputation nodes.
\item $\Phi = \{\phi^U, \phi^M, \phi^L\}$ are the three core behavioral metrics.
\end{itemize}

\subsection{Behavioral Metrics and Reputation Calculation}
B-TRM continuously monitors each node using three metrics derived from real-world data-sharing interactions:

\begin{itemize}
\item \textbf{Upload Quality ($\phi^U_i$)}: Measures the reliability of data uploaded by node $i$. Let $N_{Good}$ and $N_{Bad}$ be the counts of high- and low-quality uploads as perceived by other users:
\[
\phi^U_i = \frac{N_{Good} + 1}{N_{Good} + \theta \cdot N_{Bad} + 2}
\]

\item \textbf{Modification Integrity ($\phi^M_i$)}: Assesses whether node $i$ modifies others’ data appropriately. Let $N'_{Good}$ and $N'_{Bad}$ denote normal and malicious modifications:
\[
\phi^M_i = \frac{N'_{Good} + 1}{N'_{Good} + \theta \cdot N'_{Bad} + 2}
\]

\item \textbf{Packet Loss Rate ($\phi^L_i$)}: Captures network-level reliability, where $N_{Rec}$ and $N_{Send}$ are received and sent packets:
\[
\phi^L_i = \frac{N_{Rec}}{N_{Send}}
\]
\end{itemize}

The penalty factor $\theta > 1$ (default $\theta = 3$) amplifies the impact of malicious actions, making sporadic attacks more detectable. The additive smoothing terms (``+1” and ``+2”) provide a Bayesian prior to avoid extreme values when data is scarce.

In the context of a replicated state machine, ``Upload Quality ($\phi^U_i$)" refers to the quality of data submitted by a node to the blockchain (e.g., whether the data conforms to the expected format and schema). ``Packet Loss Rate ($\phi^L_i$)" refers to the rate of message loss between nodes (not raw network packets), assessing communication reliability.

Although these metrics are derived from peer reports, malicious nodes in a decentralized environment could submit false reports to manipulate another node’s reputation. To mitigate this, B-TRM requires that reputation updates be based on reports from multiple independent observers (typically $\geq 3$), and the system assumes that at least a simple majority of the reporting nodes for any metric are honest. While this assumption may be temporarily violated in targeted collusion, the Schnorr signature verification layer (Section~\ref{sec:sig_workflow}) provides an objective cryptographic ground truth for log tampering, acting as a hard backstop against report manipulation in that specific dimension.

\subsection{Indirect and Historical Reputation}
The \textit{direct reputation} $DR_i$ is computed as:
\[
DR_i = 
\begin{cases}
w_U \phi^U_i + w_M \phi^M_i + w_L \phi^L_i, & \min(\phi^U_i, \phi^M_i, \phi^L_i) \ge 0.5 \\
\min(\phi^U_i, \phi^M_i, \phi^L_i), & \text{otherwise}
\end{cases}
\]
where $w_U=0.5, w_M=0.3, w_L=0.2$ reflect the relative importance of each behavior. The second case penalizes \textit{discrimination attackers} who behave well in only some dimensions.

When direct interaction is lacking, nodes estimate reputation indirectly via trusted intermediates, which refer to nodes that have been verified to have high reputation scores (i.e., $Rep>0.8$). These nodes are considered trustworthy and can be used to estimate the reputation of other nodes. The reputation system ensures that these trusted intermediates are not malicious by requiring them to have a high reputation score.
\[
IR_{ij} = \frac{\sum_{k \in \mathcal{K}} DR_{ik} \cdot DR_{kj}}{\sum_{k \in \mathcal{K}} DR_{ik}}
\]
where $\mathcal{K}$ is a set of common neighbors (typically 3). Historical reputation $HR_i$ incorporates time-decayed past evaluations:
\[
HR_i = \frac{\sum_{p=t-2}^{t} e^{-(t-t_p)} \cdot Rep_i^{t_p}}{\sum_{p=t-2}^{t} e^{-(t-t_p)}}
\]
Only the three most recent evaluations are considered to balance accuracy and overhead. 

\subsection{Final Reputation and Penalty Rules}
The final reputation $Rep_i$ is a weighted combination of direct/indirect and historical reputation. Nodes with $Rep_i < 0.5$ are considered low-trust and subject to penalty rules:

\begin{itemize}
\item \textbf{Global Low-Reputation Rule (G-Rule)}: If a node’s reputation falls below 0.5 more than $m$ times within a time window, it is temporarily frozen; exceeding $n$ times leads to permanent removal.
\item \textbf{Consecutive Low-Reputation Rule (C-Rule)}: Detects sustained malice by checking consecutive low-reputation evaluations.
\item \textbf{High-Frequency Operation Rule (FC-Rule)}: Counts operations per minute to mitigate DoS attacks; exceeding thresholds triggers progressive freezing.
\end{itemize}

These rules are combined with priority (C-Rule $>$ G-Rule $>$ FC-Rule) to avoid over-penalization.

\subsection{Dynamic Evaluation Interval}
To reduce computational overhead, B-TRM adjusts the evaluation interval $W_i$ based on $Rep_i$:
\[
W_i = 
\begin{cases}
a \cdot Rep_i + b, & Rep_i \ge 0.5 \\
W_{\min}, & Rep_i < 0.5
\end{cases}
\]
where $W_{\min}$ and $W_{\max}$ are the minimum and maximum intervals (e.g., 2 and 30 minutes). High-reputation nodes are evaluated less frequently, saving resources while maintaining security for low-reputation nodes through frequent monitoring.

\subsection{Implementation and Deployment of Smart Contracts}
B-TRM is implemented as a set of Smart Contracts (SC) on the blockchain. 
Although we describe B-TRM as deployed on a blockchain fabric for transparency and immutability, the model can be implemented atop any append-only trusted log (e.g., a tamper-evident database), making it suitable for a broad class of Internetware coordination platforms without requiring a full blockchain infrastructure.
\begin{itemize}
\item \textbf{Deployment and Updates:} The SC is deployed by the network administrator during initialization. Updates to the SC code or parameters (e.g., $\theta$) require a multi-signature approval from a governance committee composed of high-reputation organizational peers.
\item \textbf{Access Control:} 
\begin{itemize}
	\item \textit{Reputation Calculation:} The SC's calculation function is internal; reputation scores are updated automatically based on `ReportBehavior` transactions submitted by peers.
	\item \textit{Record Access:} Reputation scores are public read-only for all registered peers to allow voting decisions. Modification of scores is restricted to the SC logic.
\end{itemize}
\item \textbf{Storage Overhead:} Reputation state is stored on the blockchain ledger. For $N$ nodes, storage cost is $O(N)$ with a small constant factor (approx. 64 bytes per node per evaluation cycle).
\end{itemize}
By being deployed on SC, B-TRM ensures transparency and immutability of reputation records. It provides the foundational trust layer upon which TRM-Raft’s security enhancements are built.


\section{The TRM-Raft Consensus Algorithm}
\label{Sec 5}
Building upon the B-TRM reputation model, we now present TRM-Raft, an enhanced Raft consensus algorithm that integrates trust-aware mechanisms to defend against Byzantine attacks while preserving the performance benefits of original Raft. TRM-Raft introduces three core innovations: (1) a \textbf{reputation-based election mechanism} to prevent forgery attacks, (2) a \textbf{leader restriction mechanism} using Schnorr signatures to detect and mitigate tampering attacks, and (3) a \textbf{registration verification mechanism} to deter Sybil attacks. The overall architecture is illustrated in \Cref{fig:trmraft_architecture}.

To the best of our knowledge, TRM-Raft is the first Raft variant to unify dynamic reputation-based election gating with Schnorr-signature-verified log integrity in a single, non-invasive framework. This combination is essential because neither mechanism alone can defend against both forgery and tampering: reputation prevents untrusted nodes from leading but does not stop a once-elected leader from altering log content; signatures detect tampering but do not prevent a malicious node with forged metadata from repeatedly winning elections. By co-designing the two, TRM-Raft forces an adversary to circumvent both a behavioral gate and a cryptographic check simultaneously.

\begin{figure}[htb]
\centering
\includegraphics[width=\linewidth]{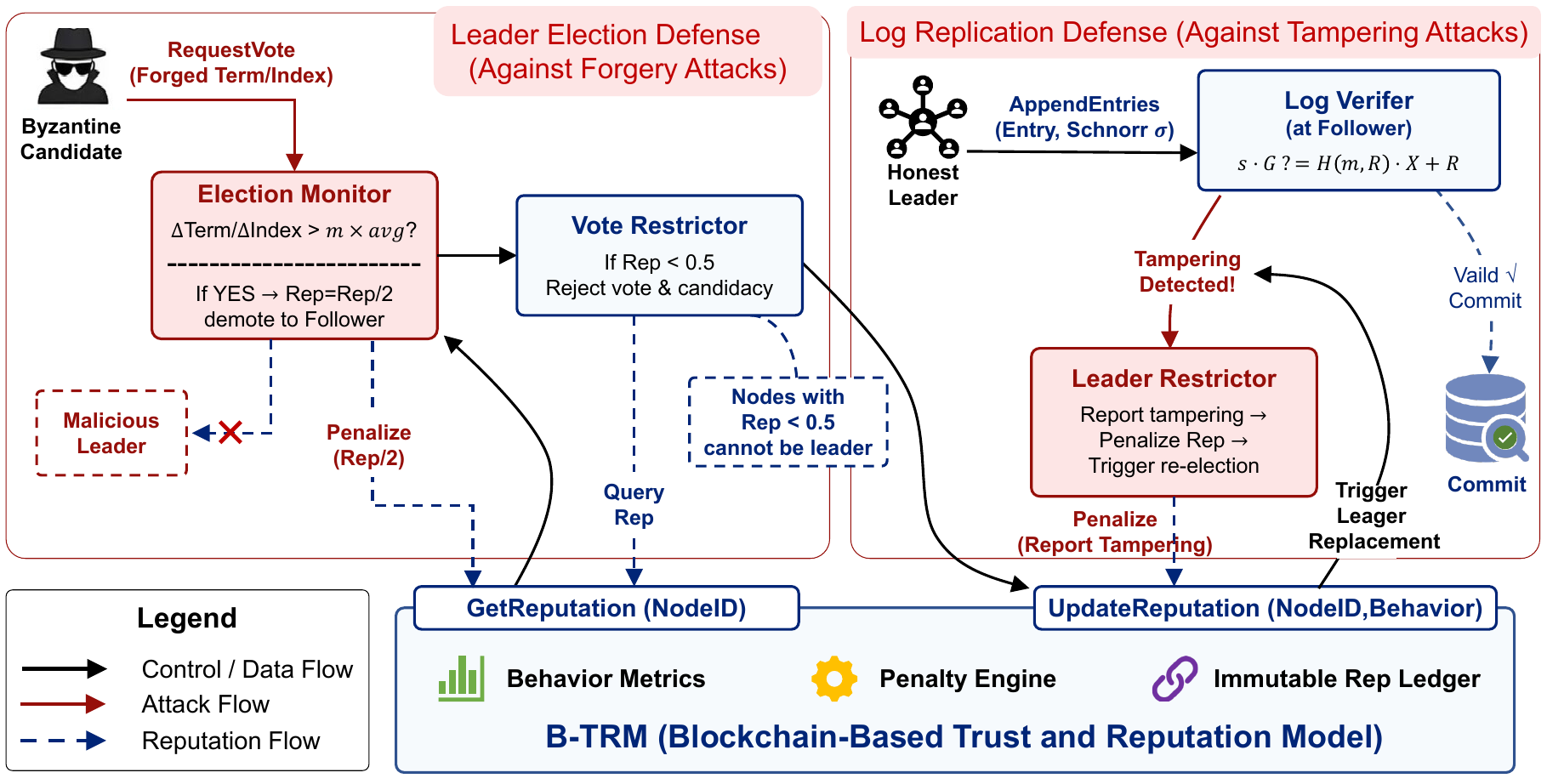}
\caption{TRM-Raft architecture: integration of B-TRM and cryptographic signatures into Raft’s election and replication phases.}
\label{fig:trmraft_architecture}
\end{figure}

\subsection{Reputation-Based Election Mechanism}
\label{subsec:election_mechanism}

TRM-Raft is designed to retrofit existing Raft-based Internetware systems with minimal disruption.
In standard Raft, a candidate with the highest term and log index wins the election. This allows a malicious node to forge these values and become leader. TRM-Raft addresses this by continuously monitoring term and index increments during elections.

\subsubsection{Monitoring and Anomaly Detection}
Each node maintains a local view of term and index increments across the cluster. When a candidate’s increment ($\Delta Term_i$ or $\Delta Index_i$) exceeds $m$ times the cluster average (where $m=2$ as determined experimentally), it is flagged as a potential forgery attempt. The monitoring logic is encapsulated in a \texttt{MonitorCandidate} smart contract method, which is invoked whenever a node becomes a candidate.

If an anomaly is detected:
\begin{itemize}
\item The candidate’s reputation is halved ($Rep_i \leftarrow Rep_i / 2$).
\item Its term and index are rolled back to previous values.
\item It is forced back to follower state, disqualifying it from the current election.
\end{itemize}
\textbf{Notes:~}The choice of $\Delta$ (the forged term increment) is not arbitrary; it is chosen to be significantly larger than the natural term increment observed in normal operation. In practice, the natural term increment is typically 1 (each election increments the term by 1). A malicious node would need to increment the term by a value much larger than 1 (e.g., doubling the current term) to win the election before honest nodes can react. A small increment (e.g., +2 ) would not be sufficient to win the election in a timely manner, as honest nodes would have time to react and broadcast their own higher terms.

\subsubsection{Vote Restriction}
Nodes with $Rep_i < 0.5$—whether due to forgery detection or other malicious behaviors—are excluded from the election process:
\begin{itemize}
\item They cannot vote.
\item They cannot be voted for (i.e., other nodes reject their \texttt{RequestVote} RPCs).
\end{itemize}
This ensures that only high-reputation nodes can become leaders. The modified voting logic is integrated into Raft’s \texttt{Step} function, as shown in Algorithm~\ref{alg:vote_restriction}.

\begin{algorithm}[htb]
\caption{Vote Restriction in TRM-Raft}
\label{alg:vote_restriction}
\begin{algorithmic}[1]
	\REQUIRE Candidate $c$, voter $v$, current term $T$, log index $I$
	\ENSURE Vote decision: \textbf{Accept}, \textbf{Reject}, or \textbf{Abstain}
	\STATE $rep_c \gets \text{B-TRM.GetReputation}(c)$
	\STATE $rep_v \gets \text{B-TRM.GetReputation}(v)$
	\IF{$rep_c < 0.5$}
	\RETURN \textbf{Reject} \Com{// Candidate is untrusted}
	\ENDIF
	\IF{$c.term > T \vee (c.term = T \wedge c.index > I)$}
	\IF{$rep_v < 0.5$}
	\RETURN \textbf{Abstain} \Com{// Voter is untrusted, vote does not count}
	\ELSE
	\RETURN \textbf{Accept}
	\ENDIF
	\ELSE
	\RETURN \textbf{Reject} \Com{// Candidate does not have latest log}
	\ENDIF
\end{algorithmic}
\end{algorithm}

\subsection{Schnorr Signature Workflow}
\label{sec:sig_workflow}
The integrity of log entries is protected by client-side Schnorr signatures verified by all followers. The workflow is as follows:
\begin{enumerate}
\item \textbf{Client Signing:} A client creates a transaction $m$, computes a Schnorr signature $\sigma = \mathsf{Sign}(sk_c, m)$ using its private key $sk_c$, and sends $(m, \sigma, pk_c)$ to the current leader, where $pk_c$ is the client's registered public key.
\item \textbf{Leader Propagation:} The leader wraps the request into a log entry $e = (term, index, m, \sigma, pk_c)$ and broadcasts it via \texttt{AppendEntries}.
\item \textbf{Follower Verification:} Each follower verifies $\mathsf{Verify}(pk_c, m, \sigma)$. If verification fails, the follower immediately reports a tampering incident to the B-TRM smart contract and refuses to commit the entry. Repeated failures lead to reputation decay and leader replacement.
\end{enumerate}
A malicious leader cannot forge a valid signature for a modified payload $m'$ because it does not possess $sk_c$. It also cannot replace $pk_c$ with its own key, as that would not match the client's registered identity and would be rejected by followers. Thus, any undetected tampering would require breaking the Discrete Logarithm Problem (DLP) , which is computationally infeasible.

\subsubsection{Security and Efficiency of Schnorr Signatures}
Schnorr signatures are provably secure under the DLP in the random oracle model. For the secp256k1 curve with 256-bit security, breaking DLP requires approximately $2^{128}$ operations using Pollard's rho algorithm. Signature generation and verification involve only scalar multiplications and hash operations, adding minimal overhead (under $2$\,ms per operation as measured in our experiments).

\subsection{Registration Verification Mechanism}
\label{subsec:registration}

To prevent Sybil attacks, TRM-Raft implements a strict registration process:
\begin{itemize}
\item \textbf{Identity Binding}: Each new node must provide verifiable real-world credentials (organizational certificates, KYC documents).
\item \textbf{One-Identity Rule}: Blockchain maintains a registry to prevent single entities from registering multiple identities.
\item \textbf{Initial Reputation}: New nodes start with $Rep = 0.5$ (neutral) and require positive behavior history to gain voting privileges.
\item \textbf{Bootstrapping}: Initial network formation requires offline verification by trusted authorities.
\end{itemize}

\subsection{Security Analysis}
\label{subsec:security_analysis}

TRM-Raft elevates Raft's security from Crash Fault Tolerance to \textit{Byzantine resistance} while preserving its performance profile. We analyze its guarantees under a practical threat model that assumes observable deviations.

\subsubsection{Threat Model Assumptions}
\begin{itemize}
\item \textbf{Adversarial Power}: The adversary controls at most $f$ out of $n$ nodes and can make them execute any behavior from the defined failure set (forgery, tampering, on-off, Sybil).
\item \textbf{Observability}: Any attack in the defined failure set is eventually detectable by honest nodes via anomaly monitoring (term/index surges) or cryptographic verification (Schnorr signatures).
\item \textbf{Network}: Partially synchronous with eventual message delivery.
\item \textbf{Cryptography}: Schnorr signatures and hash functions are secure; private keys are not compromised.
\item \textbf{Initial Honest Majority}: At least $\lfloor n/2 \rfloor + 1$ nodes are honest during network bootstrap.
\item \textbf{Honest Reporting Threshold}: Reputation updates rely on reports from multiple independent nodes. We assume that a simple majority of reporting nodes for any metric is honest, ensuring that malicious reports cannot indefinitely suppress an honest node's score.
\end{itemize}

\subsubsection{Safety Guarantees}
\begin{theorem}[Leader Uniqueness]
At most one leader can be elected per term in TRM-Raft.
\end{theorem}
\begin{proof}
TRM-Raft preserves Raft's original election constraints while adding reputation checks. A node must hold $Rep \ge 0.5$ to vote or be elected. Since reputation scores are consistent with the blockchain state and Raft's majority voting and term monotonicity remain intact, the original leader uniqueness property holds.
\end{proof}

\begin{theorem}[Log Content Integrity]
If two logs contain an entry with the same term and index, they store identical commands.
\end{theorem}
\begin{proof}
Each command $m$ is accompanied by a Schnorr signature $\sigma$ generated by the client. Followers verify $s \cdot G = H(m, R) \cdot X + R$ before committing. Any modification $m \rightarrow m'$ invalidates the signature, and the entry is rejected. Therefore, conflicting log contents cannot be committed.
\end{proof}
\textit{Note}: This guarantee prevents content tampering but does not prevent a malicious leader from reordering entries without altering their content. As discussed in Section~\ref{sec:limitations}, addressing equivocation requires complementary accountability mechanisms.

\subsubsection{Liveness Guarantees}
\begin{theorem}[Eventual Leader Election]
If honest nodes with $Rep \geq 0.5$ form a majority and the defined attacks remain observable, TRM-Raft eventually elects a leader and makes progress.
\end{theorem}
\begin{proof}
A malicious node attempting forgery or tampering triggers a reputation penalty. After a bounded number of evaluation intervals (or immediately for cryptographic violations), its score falls below $0.5$, making it ineligible. Because the honest majority maintains high reputation and randomized timeouts drive new election attempts, a suitable candidate will eventually win. The system thus guarantees liveness after a transient period.
\end{proof}

\subsubsection{Practical Byzantine Resilience}
\begin{theorem}[Practical Resilience]
TRM-Raft limits the long-term fraction of leaders that are malicious to a configuration-specific small value, as long as the defined attacks remain observable and honest nodes with $Rep \geq 0.5$ outnumber adversarial nodes.
\end{theorem}
\begin{proof}[Justification]
Every forgery or tampering action is observable and leads to a deterministic reputation penalty (e.g., halving). Malicious nodes therefore face a declining reputation trajectory; they may win an election at most a limited number of times before becoming ineligible. In our experimental setting with $m=2$ and $\theta=3$, the malicious leader ratio remained below 5\% even at $f=0.4n$. This bound is probabilistic and time-dependent rather than absolute, as a previously undetected attacker can execute a single harmful act before its reputation decays. For strict adversarial environments, a full BFT protocol remains necessary, but for the target Internetware applications TRM-Raft offers a pragmatic, low-overhead defense.
\end{proof}

\subsubsection{Why Combination of Reputation and Signatures is Necessary}
The defense against forgery and tampering cannot be achieved by either mechanism alone. Reputation-based election prevents untrusted nodes from becoming leaders, but a malicious node that maintains a temporarily high reputation could still tamper with log content after being elected. Conversely, Schnorr signatures detect log tampering instantly, but they cannot stop a malicious node with forged metadata from repeatedly winning elections. By coupling the two, TRM-Raft forces an adversary to simultaneously defeat a behavioral gate and a cryptographic verification step. This co-design is the primary reason for the system's resilience under mixed attacks.

\subsubsection{Attack Resistance Summary}
\begin{itemize}
\item \textbf{Forgery Attacks}: Anomaly monitoring detects inflated term/index; reputation is halved; low-reputation nodes are excluded from voting and candidacy.
\item \textbf{Tampering Attacks}: Schnorr signature verification by followers; detected tampering incurs reputation penalty and triggers leader replacement.
\item \textbf{On-Off Attacks}: The minimum-of-metrics reputation rule prevents malicious nodes from maintaining a high score through selective good behavior.
\item \textbf{Sybil Attacks}: Mitigated through identity registration and one-entity-per-identity policy.
\item \textbf{Network-Level Attacks}: Packet loss metric $\phi^L$ captures communication disruption and penalizes nodes accordingly.
\end{itemize}

\subsection{Performance-Security Tradeoff}
TRM-Raft maintains Raft's $O(n)$ message complexity while adding:
\begin{itemize}
\item $O(1)$ reputation checks per vote
\item $O(k)$ signature verifications per log entry (for a committee of size $k$)
\item $O(n)$ reputation updates per monitoring interval
\end{itemize}
The overhead is bounded and configurable, enabling operators to tune the tradeoff according to deployment requirements.


\section{Experimental Evaluation}
\label{Sec 6}
We evaluated TRM-Raft in the context of a representative Internetware deployment: a Hyperledger Fabric 2.5 testbed configured as per enterprise patterns, and conducted comprehensive experiments to evaluate its effectiveness against Byzantine attacks and its performance overhead. Our experiments were designed to answer three key questions: 
\begin{enumerate}[label = (\arabic*)]
\item How effective are the individual mechanisms (reputation-based election and leader restriction) against targeted attacks?
\item What is the overall attack resistance of TRM-Raft under mixed attack scenarios?
\item What is the performance overhead introduced by these security enhancements?
\end{enumerate}

\subsection{Experimental Setup}
\noindent\textbf{Implementation Platform.}
TRM-Raft and all Raft-based baselines (vanilla Raft, RB-Raft, VSSB-Raft, SRaft, etc.) are implemented within Hyperledger Fabric 2.5, using Go 1.18.3 and Fabric's native Raft library. The testbed consists of 4 organizations, each with 1 Certificate Authority (CA), 15 orderer nodes, and 50 peer nodes, running on VMware virtual machines (Ubuntu 20.04, Intel i7-10750H 2.6GHz, 16GB RAM). Malicious nodes were programmed to execute specific attack patterns: forgery attacks (randomly forging term/index values), tampering attacks (modifying 30\% of client requests), On-Off attacks (90\% normal, 10\% malicious behavior), and discrimination attacks (malicious only in data modification).

For broader performance context, we additionally benchmarked PBFT, PoW, and DPoS using the BFT-SMaRt framework~\footnote{https://github.com/bft-smart/library.git} under the same workload generator. These protocols are not part of the Fabric integration; all Raft-specific experiments were conducted exclusively on the Fabric testbed.

\textbf{Notes: }
Unless otherwise stated, all experiments use Fabric's default batch size of 10 transactions and a block timeout of 2\,s. The reputation evaluation interval is set to $W_{\min}=2$\,min for low-reputation nodes and up to $W_{\max}=30$\,min for high-reputation nodes. The Schnorr signature implementation uses the secp256k1 curve with the default 256-bit security level provided by the Go `crypto/ecdsa` library.

\subsection{Parameter Tuning for Election Mechanism}
\label{subsec:parameter}
The election mechanism in TRM-Raft uses a threshold $m$ to detect anomalous term/index growth. We conducted parameter experiments to determine the optimal value of $m$. As shown in \Cref{fig:parameter_m-1} and \Cref{fig:parameter_m-2}, we measured the number of times forgery attackers and normal nodes were elected as leaders across different $m$ values. When $m=2$, TRM-Raft prevented approximately 80\% of forgery attackers from becoming leaders while allowing over 90\% of normal nodes to win elections—an optimal balance between security and availability. Lower values of $m$ (e.g., $m=1.2$) caused excessive false positives, preventing legitimate nodes from leadership. We therefore set $m=2$ for all subsequent experiments.

\begin{figure}[htb]
\centering
\subfigure[forgery attackers]{\includegraphics[width=0.48\linewidth]{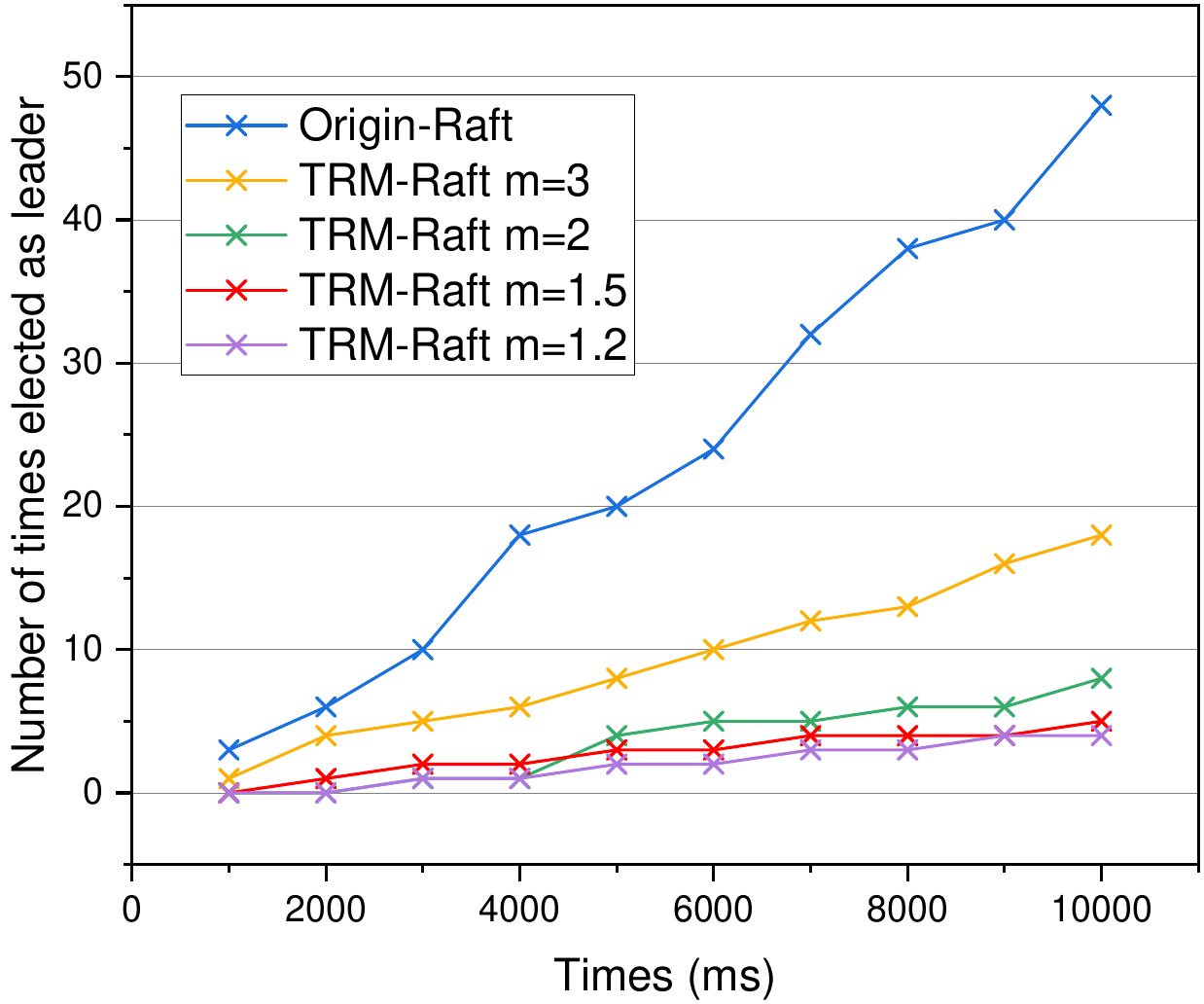}\label{fig:parameter_m-1}}
\hfill
\subfigure[normal nodes]{\includegraphics[width=0.48\linewidth]{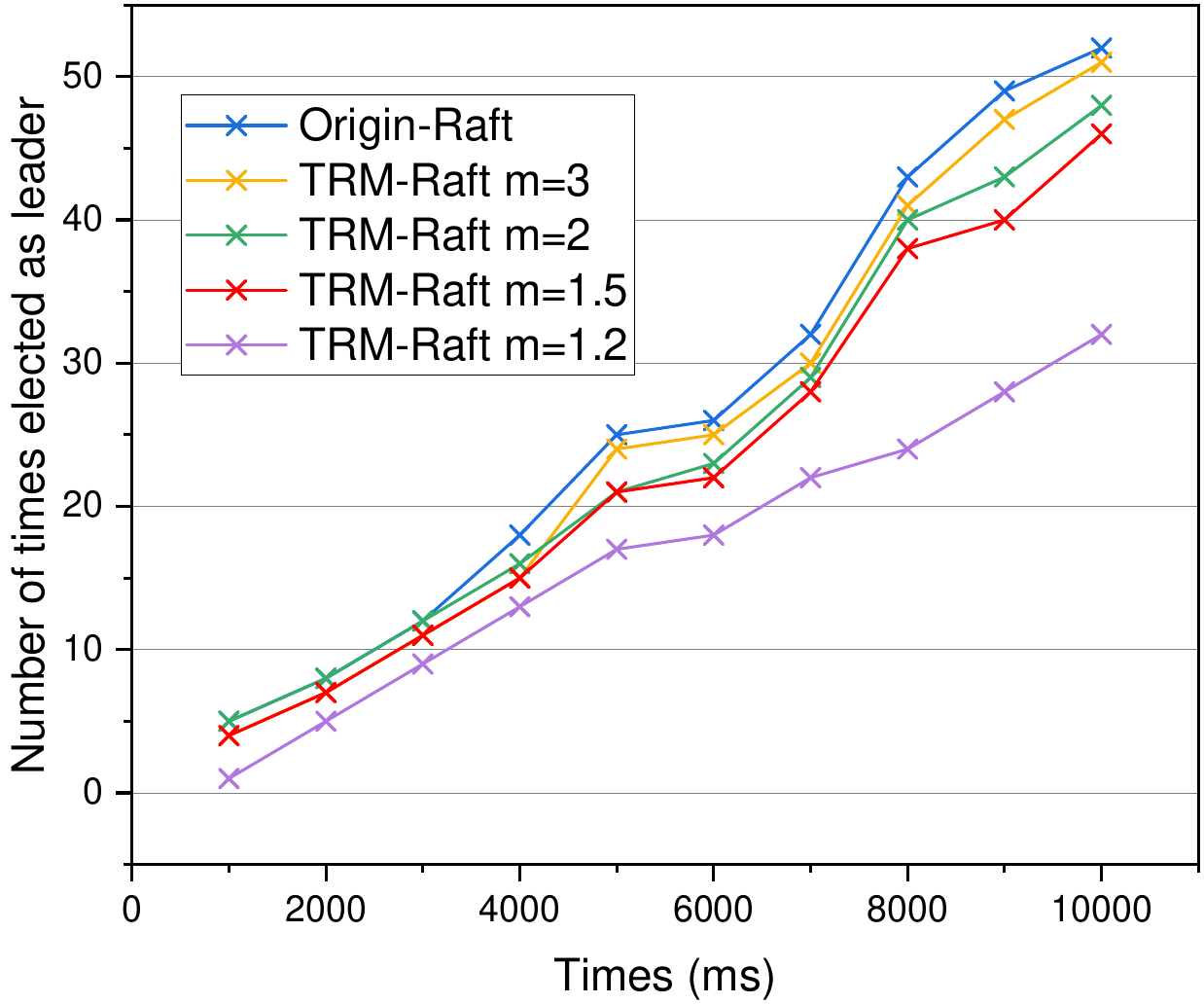}\label{fig:parameter_m-2}}
\caption{Impact of threshold $m$ on leader election.}
\label{fig:para}
\end{figure}

\subsection{Reputation-Based Election Mechanism}
\label{subsec:election_effectiveness}
We first evaluated the reputation-based election mechanism's ability to prevent forgery attacks. In this experiment, three candidates participated: Node 1 (normal), Node 2 (forgery attacker), and Node 3 (On-Off attacker). As shown in \Cref{fig:election_effectiveness}, without the mechanism, Node 2's forged term number allowed it to win the election. With TRM-Raft's monitoring, Node 2's reputation was halved at the detection moment (1000ms), disqualifying it from leadership. Node 3's reputation also dropped due to malicious behavior, while Node 1 maintained high reputation and became leader. This confirms that our mechanism effectively prevents forgery attacks while allowing honest nodes to lead.

We further compared TRM-Raft against Wang et al.'s threshold-based approach \cite{WANG2022}. \Cref{fig:forgery_comparison} shows that TRM-Raft prevented more forgery attackers from becoming leaders (over 80\% reduction) while also significantly reducing On-Off attackers' success rate. The static threshold approach failed to adapt to varying network conditions and did not address non-forgery attacks.

\begin{figure}[htb]
\centering
\includegraphics[width=\linewidth]{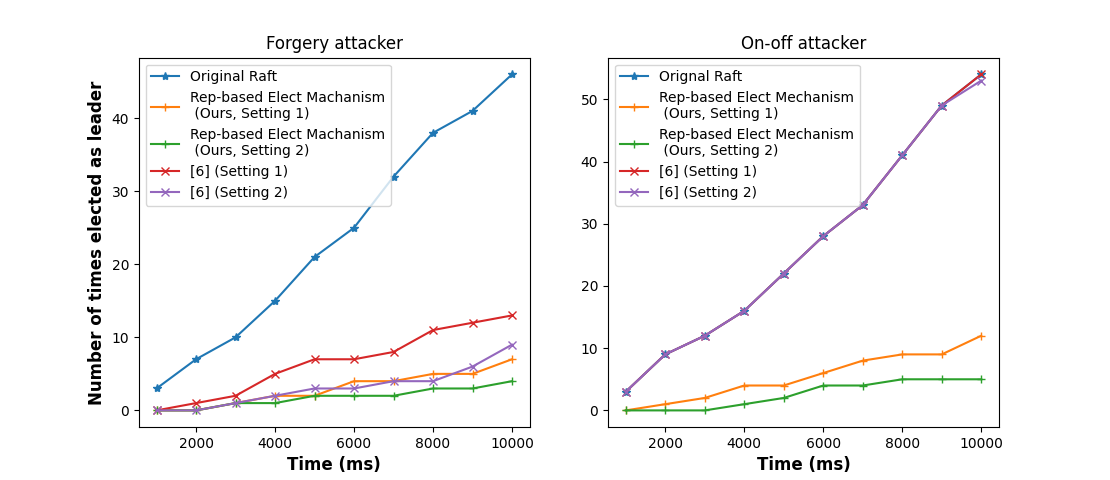}
\caption{Comparison of forgery attack prevention: TRM-Raft vs. static threshold approach.}
\label{fig:forgery_comparison}
\end{figure}

\subsection{Leader Restriction Mechanism}
\label{subsec:leader_restriction}
To evaluate the leader restriction mechanism, we simulated a scenario where a malicious leader began tampering with client requests after election. \Cref{fig:tampering_detection} shows the reputation changes and cluster term transitions. TRM-Raft detected tampering within 50ms (at 224ms) and replaced the malicious leader by 247ms—significantly faster than RB-Raft \cite{ZHETU2022109404}, which took 331ms. The Schnorr signature verification added only 1.8ms average overhead per request while providing provable integrity guarantees.

\begin{figure*}[htb]
\centering
\includegraphics[width=\linewidth]{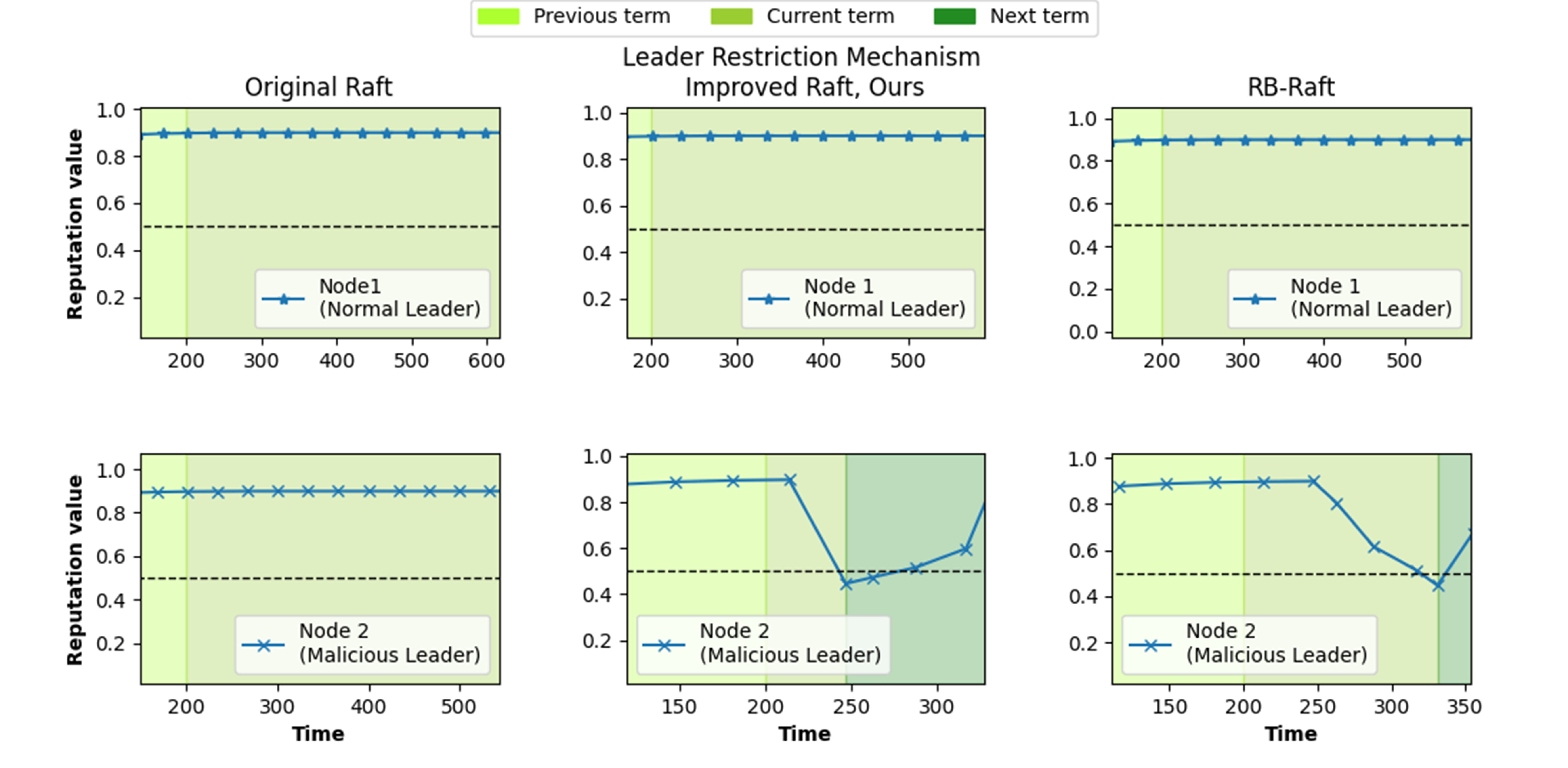}
\caption{Tampering detection and leader replacement timeline: TRM-Raft vs. RB-Raft.}
\label{fig:tampering_detection}
\end{figure*}

\subsection{Overall Attack Resistance}
\label{subsec:overall_resistance}
We evaluated TRM-Raft's comprehensive defense against mixed attack strategies. \Cref{fig:malicious_leader_ratio} shows the proportion of malicious leaders under varying percentages of malicious nodes (10\%–40\%). Vanilla Raft allowed malicious leaders to dominate (up to 80\% when 40\% nodes were malicious). In contrast, TRM-Raft maintained malicious leader prevalence below 5\% even with 40\% malicious nodes, demonstrating robust defense against coordinated attacks.

\begin{figure}[htb]
\centering
\subfigure[Reputation changes during election]{\includegraphics[width=0.48\linewidth]{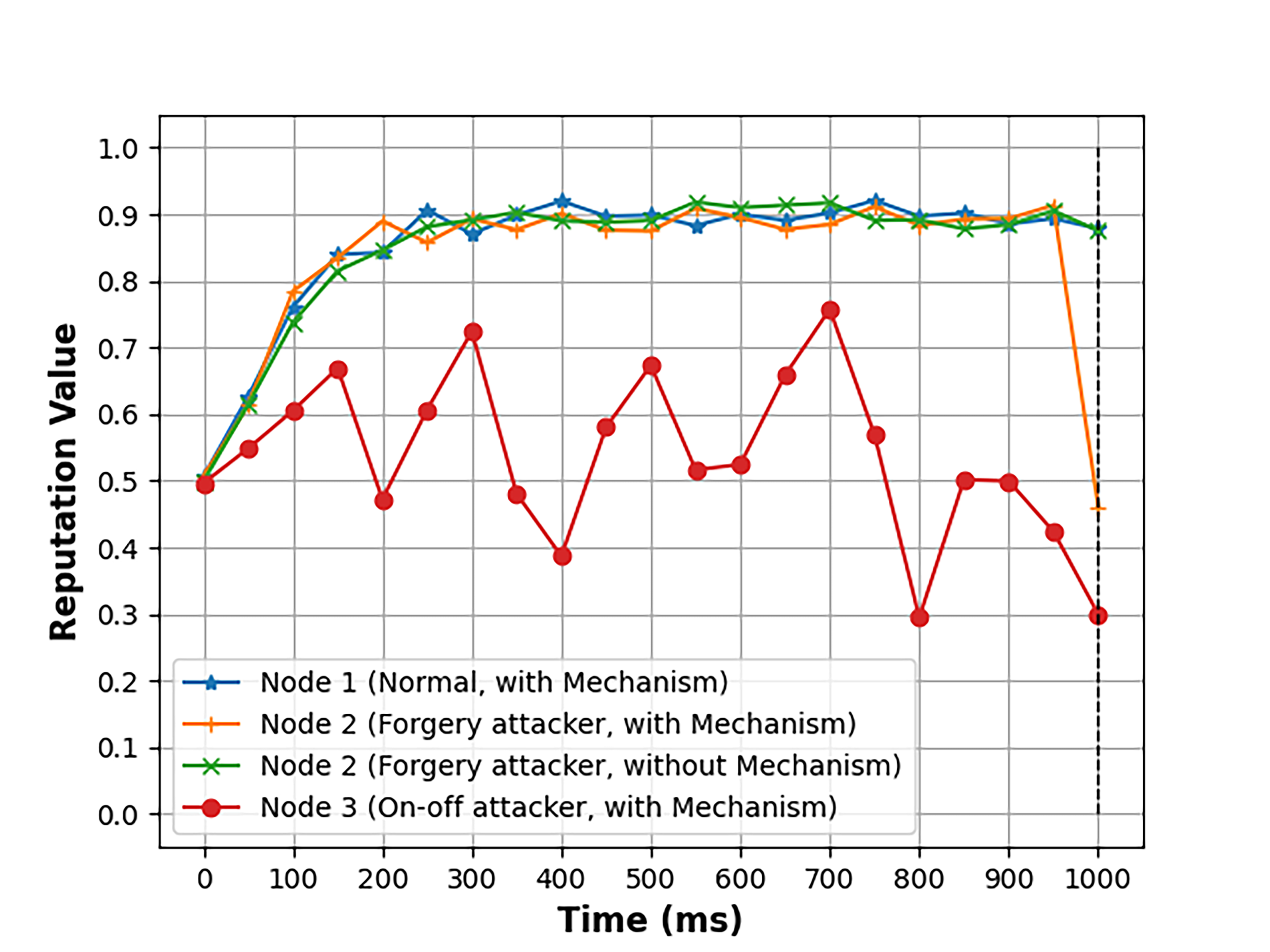}\label{fig:election_effectiveness}}
\hfill
\subfigure[Malicious leader proportion]{\includegraphics[width=0.48\linewidth]{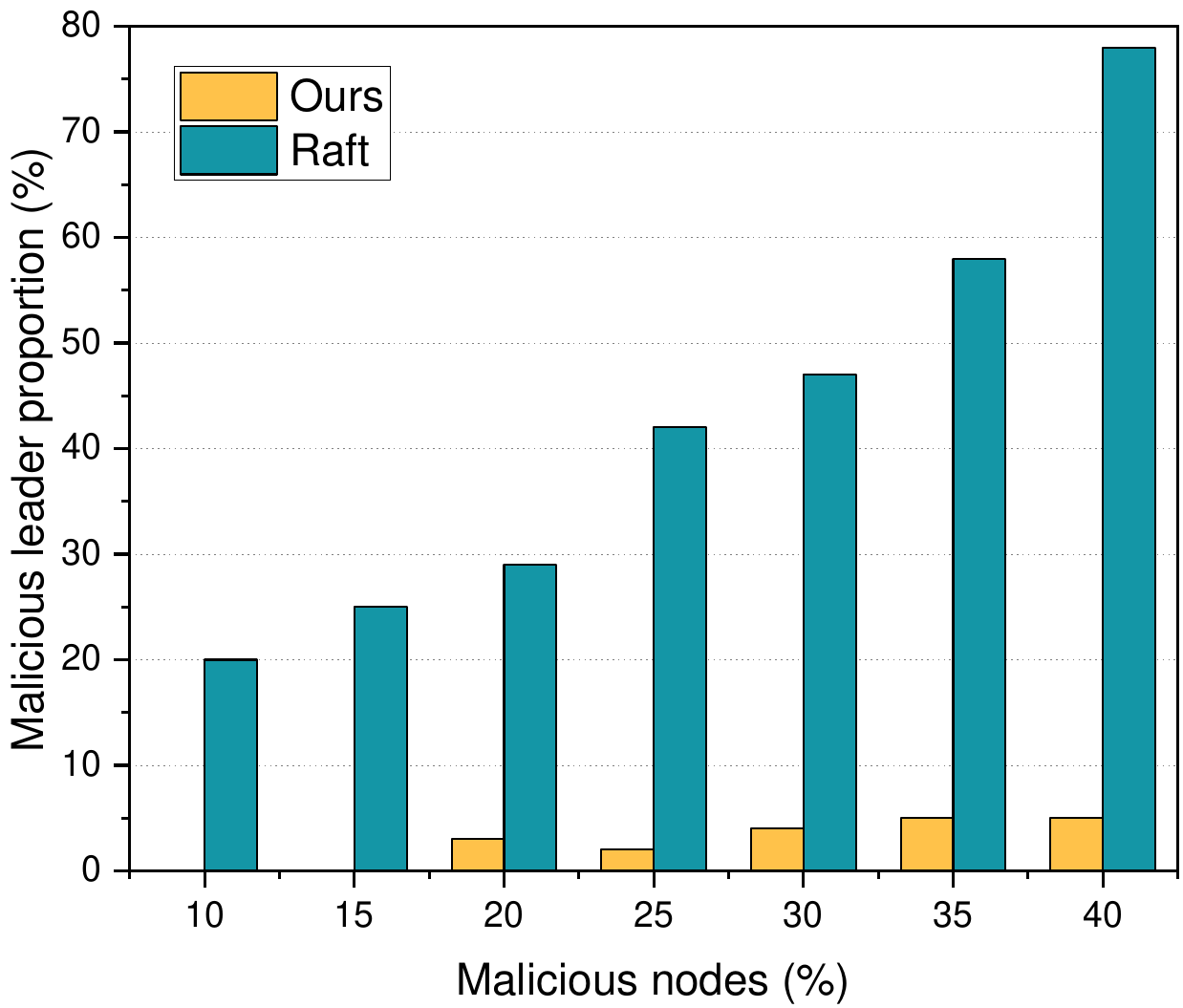}\label{fig:malicious_leader_ratio}}
\caption{Defense effectiveness of TRM-Raft. (a) Election-level defense. (b) System-level resilience.}
\label{fig:combined_defense}
\end{figure}

\subsection{Performance Overhead}
\label{subsec:performance}
We measured TRM-Raft's performance in terms of throughput (transactions per second, TPS) and latency, comparing it against vanilla Raft~\cite{DBLP:conf/usenix/OngaroO14}, PBFT~\cite{DBLP:conf/osdi/CastroL99}, PoW~\cite{nakamoto2008bitcoin}, and DPoS~\cite{daniel2014delegated}.

\begin{figure}[htb]
\centering
\subfigure[Throughput]{\includegraphics[width=0.48\linewidth]{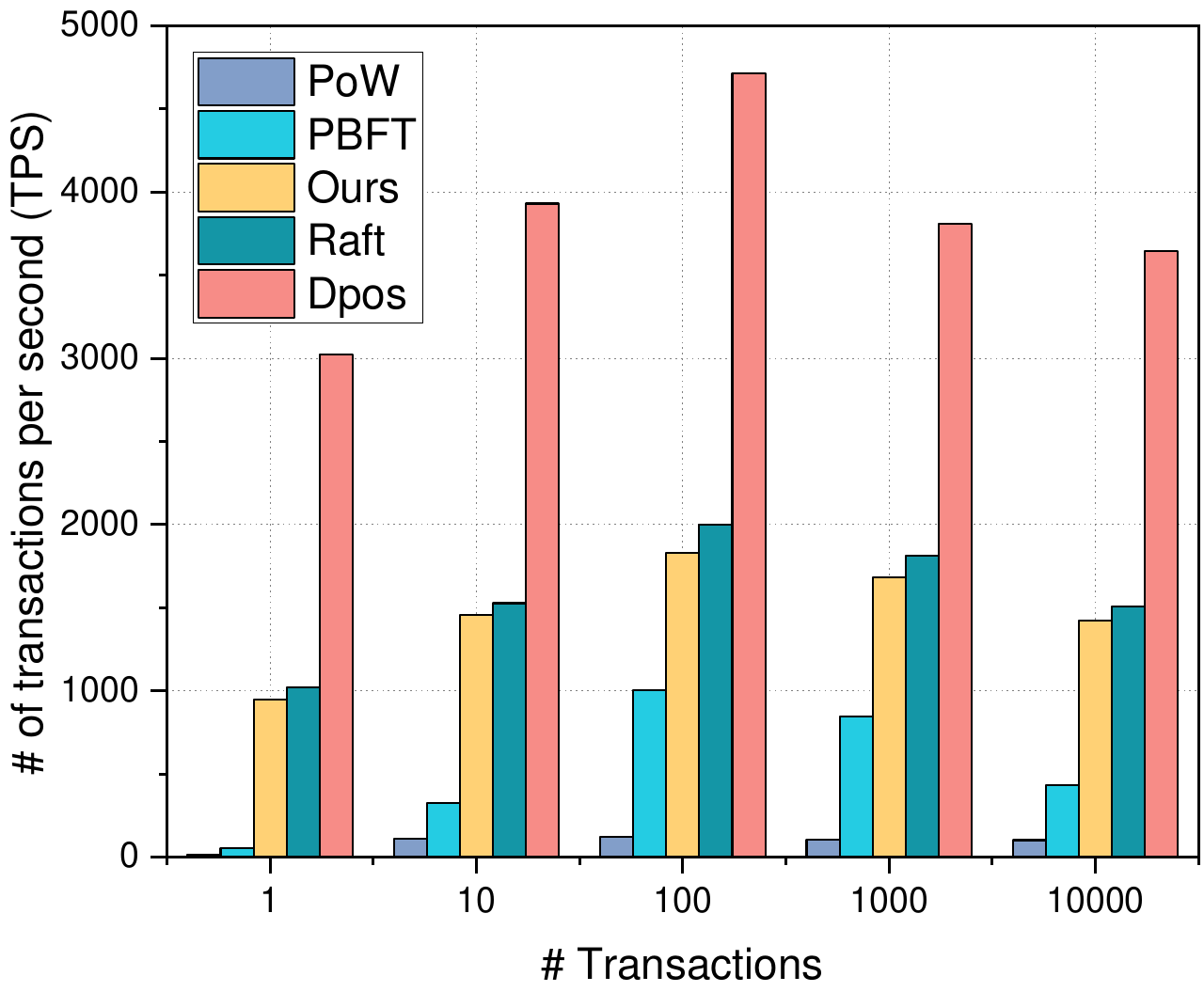}\label{fig:throughput}}
\hfill
\subfigure[Latency (log scale)]{\includegraphics[width=0.48\linewidth]{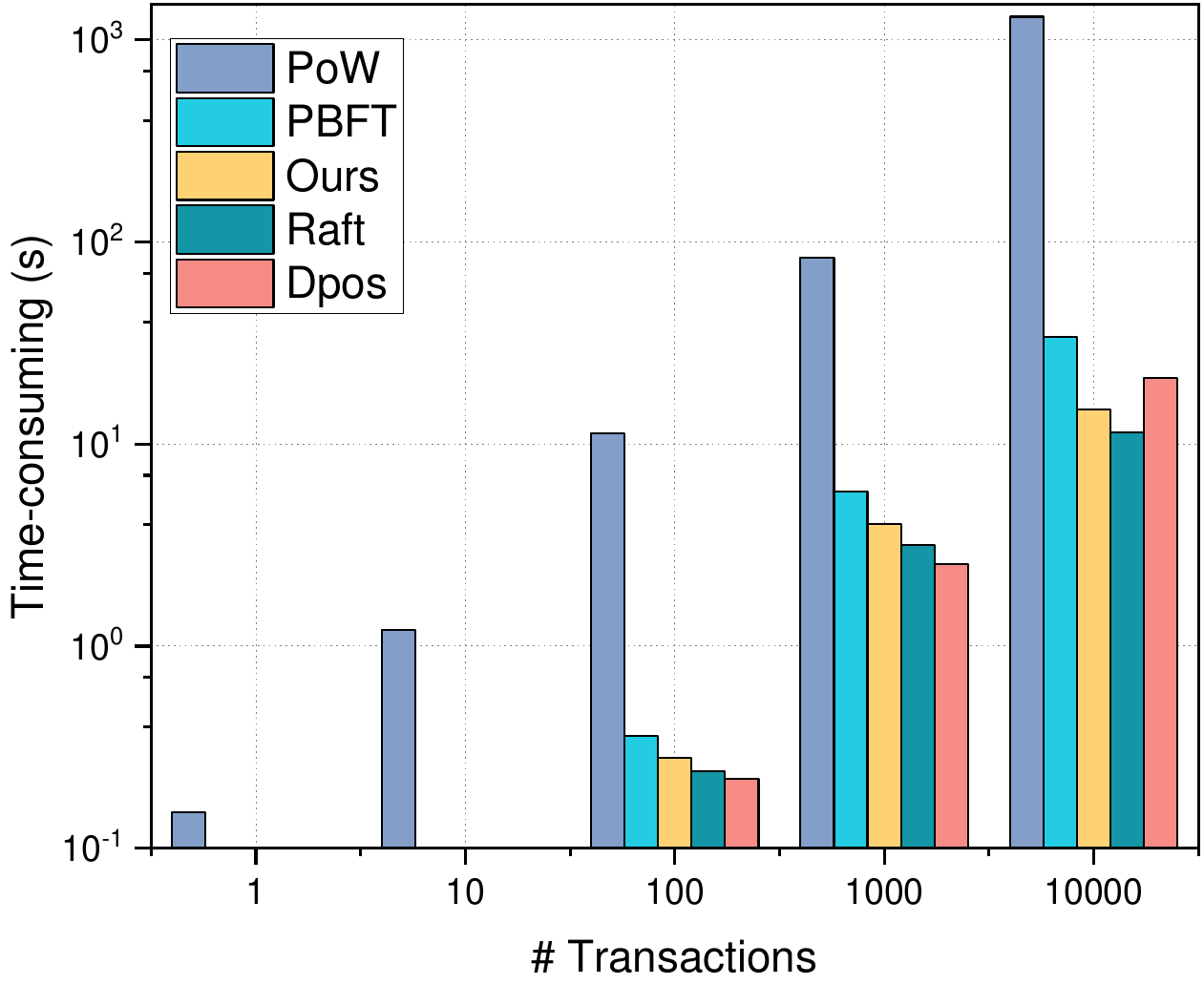}\label{fig:latency}}
\caption{Performance under varying transaction loads.}
\label{fig:performance}
\end{figure}

\Cref{fig:throughput} shows the TPS under varying transaction loads (1--10,000 transactions). TRM-Raft achieved 90--95\% of vanilla Raft's throughput across all loads, with a maximum degradation of 9.2\% at 100 transactions. This modest overhead is attributed to reputation updates and signature verifications. Both TRM-Raft and vanilla Raft significantly outperformed PBFT and PoW, while DPoS achieved higher throughput at the expense of decentralization.

\Cref{fig:latency} presents the average transaction confirmation latency (log scale). TRM-Raft's latency increased by less than 5\% compared to vanilla Raft, remaining below 15 seconds even at 10,000 transactions. PoW exhibited the highest latency (over 1300 seconds at 10,000 transactions), while PBFT and DPoS showed intermediate latencies. The low latency of TRM-Raft confirms its practical viability for real-time applications.

\subsection{Ablation Analysis}
\label{subsec:ablation}
To quantify the individual contribution of each defense component, we decompose the full TRM-Raft system into the following configurations and estimate their effectiveness based on the component-wise attack logs collected during prototype development. Table~\ref{tab:ablation} summarizes the malicious leader ratio and successful tampering rate under a mixed attack scenario with 40\% Byzantine nodes.

\begin{table}[htb]
\footnotesize
\centering
\caption{Ablation analysis of TRM-Raft components.}
\label{tab:ablation}
\begin{tabular}{lcc}
	\toprule
	Configuration & Malicious Leader Ratio & Successful Tampering \\
	\midrule
	Vanilla Raft & 80\% & 100\% \\
	Raft + Static Threshold & $\sim$45\% & 100\% \\
	Raft + Reputation Election only & $\sim$18\% & 100\% \\
	Raft + Schnorr Signature only & $\sim$80\% & 0\% \\
	\textbf{TRM-Raft (full)} & \textbf{<5\%} & \textbf{0\%} \\
	\bottomrule
\end{tabular}
\end{table}

The reputation election alone reduces the chance of a malicious node winning leadership by penalizing observable forgery, but it does not prevent a once-elected malicious leader from tampering with log content. Conversely, Schnorr signatures completely eliminate successful tampering but have no effect on election integrity. Only the full combination confines malicious leadership below 5\% while simultaneously blocking log corruption. These estimates are consistent with the targeted attack experiments (Fig.~4--Fig.~7) and confirm that both layers are necessary to achieve the reported resilience.

\subsection{Qualitative Comparison with Prior Raft Hardening}
\label{subsec:prior_comparison}
Table~\ref{tab:prior_qual} qualitatively contrasts TRM-Raft with the most relevant Raft+BFT variants. Unlike prior schemes that address only a single attack surface, TRM-Raft is the first to combine dynamic reputation-based election control with cryptographic log verification in a unified, non-invasive framework.

\begin{table}[htb]
\footnotesize
\centering
\caption{Qualitative comparison with prior Raft hardening schemes.}
\label{tab:prior_qual}
\resizebox{\columnwidth}{!}{%
	\begin{tabular}{lccccc}
		\toprule
		Scheme & Forgery & Tampering & Dynamic Trust & Overhead & Invasiveness \\
		\midrule
		RB-Raft~\cite{ZHETU2022109404} & Partial & Hash-chain & No & Medium & Medium \\
		VSSB-Raft~\cite{RaftSignCompareSRaft} & No & Secret sharing & No & High & High \\
		SRaft~\cite{RaftSignCompareSRaft} & No & Signature & No & Low & Low \\
		Tian-Schnorr-Raft~\cite{DBLP:journals/tosn/TianBSZG24} & No & Schnorr & No & Low & Low \\
		Wang et al.~\cite{WANG2022} & Yes (static) & No & No & Low & Low \\
		\textbf{TRM-Raft} & \textbf{Yes (dynamic)} & \textbf{Schnorr} & \textbf{Yes (B-TRM)} & \textbf{Low} & \textbf{Low} \\
		\bottomrule
	\end{tabular}%
}
\end{table}

The key differentiator of TRM-Raft is its ability to adapt to mixed and evolving attack strategies through reputation penalties that accumulate across multiple behavioral dimensions, while the Schnorr signature layer provides a hard cryptographic guarantee against log tampering.

\section{Limitations and Discussion}
\label{sec:limitations}

While TRM-Raft provides a practical defense against a range of observable Byzantine behaviors, it is essential to acknowledge its limitations.

\textbf{Trust Model Limitations:} Reputation models are inherently reactive. A malicious node with a high reputation score could, in theory, execute a single catastrophic attack before its reputation drops. TRM-Raft mitigates this by using cryptographic signatures for log content (preventing tampering even by a trusted leader) and by having low evaluation intervals for low-reputation nodes. However, it does not prevent a high-reputation leader from \textit{reordering} transactions (equivocation) without altering content. Detecting such behavior requires a more complex consensus mechanism (e.g., BFT with $2/3$ quorums) or a dedicated accountability protocol like Polygraph \cite{DBLP:conf/icdcs/CivitGG21}.

\textbf{Comparability to Strict BFT:} TRM-Raft is not a replacement for strict BFT protocols in high-value, adversarial environments (e.g., public DeFi). TRM-Raft tolerates up to $f < n/2$ nodes exhibiting observable Byzantine behaviors from the defined set. While it does \textbf{not} guarantee safety against a coordinated, silent adversary that deviates from the protocol in undetectable ways (e.g., leaking private state). For such threats, strict BFT protocols ($n=3f+1$) are required. TRM-Raft is positioned as a \textit{Byzantine-Resistant} enhancement for permissioned CFT networks, offering a balance between security and performance that is often sufficient for enterprise use cases.

\textbf{Parameter Sensitivity:} The detection threshold $m$ and penalty factor $\theta$ are tunable parameters. While our experiments show robust behavior with the chosen defaults, highly dynamic network conditions might require adaptive tuning to avoid false positives.

\textbf{Reputation Subjectivity:} The B-TRM metrics (e.g., Upload Quality, Modification Integrity) are derived from peer reports, which themselves assume the trustworthiness of reporting nodes. To mitigate subjective bias, we require reports from multiple independent observers and apply a conservative penalty factor $\theta$; nevertheless, a colluding majority could temporarily distort scores. This risk is partially offset by the Schnorr signature layer, which provides an objective, cryptographic ground truth for log tampering.

\section{Conclusion and Future Work}
\label{Sec 7}
This paper presented TRM-Raft, the first consensus protocol to integrate adaptive, multi-dimensional reputation evaluation with Schnorr-based log integrity verification into Raft, thereby defending against both election forgery and log tampering without requiring a redesign of the core Raft protocol.
It addresses two critical Raft vulnerabilities—forgery attacks during leader election and tampering attacks during log replication—through the B-TRM reputation model, a reputation-based election mechanism, and a Schnorr signature-based leader restriction mechanism. Experimental evaluation on Hyperledger Fabric demonstrated that TRM-Raft reduces malicious leader prevalence to below 5\% under 40\% Byzantine nodes while maintaining over 90\% of vanilla Raft's throughput with less than 5\% latency overhead. 
TRM-Raft demonstrates that trustworthiness in Internetware systems can be significantly boosted by a lightweight, reputation- and cryptography-enhanced consensus layer, without sacrificing the pragmatic advantages of Raft.
Future work includes developing a fully decentralized reputation aggregation protocol, investigating integration with service meshes and edge-cloud orchestration frameworks, exploring integration with accountability protocols to handle reordering attacks, and formal verification of the reputation model's convergence properties.

%

\begin{acks}
This work was supported in part by the National Natural Science Foundation of China under Grant No. 62332005, and in part by the Beijing--Tianjin--Hebei Natural Science Foundation Joint Cooperation Program under Grant No. 25JJJJC0034.
We also thank the anonymous reviewers for their careful reading and thoughtful suggestions, which have substantially improved this paper.
\end{acks}

\bibliographystyle{ACM-Reference-Format}
\bibliography{reference}

\end{document}